\documentclass[12pt]{article}
\usepackage[iso-8859-7]{inputenc}
\usepackage{amsfonts}
\usepackage{graphicx}
\usepackage{url}
\usepackage{amsthm}
\usepackage{color}
\usepackage{braket}
\usepackage{extarrows}
\usepackage{titling}

\newtheorem{lemma}{Lemma}

\newtheorem{proposition}{Proposition}
\newtheorem{protocol}{Protocol}

\DeclareMathOperator{\Tr}{Tr}

\def\ppmod{\!\!\!\!\!\pmod}

\def\poly{poly}
\def\span{span}
\def\O{\textbf{\textit{O}}}

\date{}

\begin{document}
\title{Quantum walks public key\\ cryptographic system}

\author{
C. Vlachou
\hspace*{.1cm}
J. Rodrigues
\hspace*{.1cm}
P. Mateus\\
\hspace*{.1cm}
 N. Paunkovi\'c
\hspace*{.1cm}
A. Souto
\\ \ \\
SQIG - Instituto de Telecomunica\c c\~oes\\
Departamento de Matem\'atica - IST - UL\\
}

\maketitle

\begin{abstract}

Quantum Cryptography is a rapidly developing field of research that benefits from the properties of Quantum Mechanics in performing cryptographic tasks. Quantum walks are a powerful model for quantum computation and very promising for quantum information processing. In this paper, we present a quantum public-key cryptographic system based on quantum walks. In particular, in the proposed protocol the public key is given by a quantum state generated by performing a quantum walk. We show that the protocol is secure and analyze the complexity of public-key generation and encryption/decryption procedures.
\end{abstract}


\section{Introduction}

Since the invention of writing, the need for secret/secure communication resulted in the development of cryptography -- the art of ``hidden communication''.  It started by using simple symbols as code words and evolved to the stage where the security is based on various mathematical hardness assumptions: a widely used RSA-based cryptographic system~\cite{riv:sha:adl:78} relies on the conjecture that factoring large numbers is not feasible using standard computers, while the alternative lattice-based public-key cryptographic system~\cite{gol:gol:hal:97} is based on the assumed difficulty of the so-called ``shortest and closest vector problems'' (also related to the well known $P\neq NP$ conjecture~\cite{fort:13}). 
With the advent of quantum computation, and in particular after the discovery of the celebrated Shor's algorithm for factoring~\cite{sho:97}, the security of most cryptographic systems currently in use became jeopardized and as a consequence the need for new cryptographic systems resilient to quantum adversaries arose. The above mentioned lattice-based cryptographic system is resilient to particular quantum adversaries that execute Shor's algorithm, as its security is not based on the factoring problem; nevertheless, it does rely on another mathematical assumption, thus leaving open the question of whether it is more than just computationally  secure.  A completely different approach is the one that takes advantage of quantum mechanics and assumes that both sender and receiver have access to a quantum computer.  The major advantage of this approach is that the security, rather than relying on mathematical/computational  hardness assumptions, is now based on the laws of quantum mechanics. 

Quantum cryptography has been among hot subjects of research in the last decades. It was firstly considered by Wiesner in the late sixties and early seventies, who introduced the notions of quantum multiplexing and money (which was only published a decade later, see~\cite{wie:83}), and further developed by Bennet and Brassard in their famous Quantum Key Distribution (QKD) BB84 protocol~\cite{ben:bra:84}. 
Subsequently, Mayers showed the unconditional security of BB84 protocol
i.e., the security that is based on the laws of (quantum) physics (see~\cite{may:01}).  Rapid development of experimental techniques resulted in the implementation of QKD, see for example~\cite{sch:07},
and nowadays one can even buy devices implementing it on the online market from Clavis QKD Platform~\cite{clavis}.

As an alternative to QKD, another way for two parties to communicate privately has been developed, using quantum public-key cryptosystems: in~\cite{oka:tan:uch:00} the authors propose a scheme based on quantum trapdoor one-way functions, and in~\cite{kaw:etal:05} a system is presented whose security is based on (quantum) computational indistinguishability of quantum states. More recently, Nikolopoulos~\cite{nik:08} presented a secure public-key cryptographic system based on single-qubit rotations. 

In this paper we propose a new quantum public-key cryptographic system in which the public keys are quantum states generated by running a quantum walk, rather than by performing single-qubit rotations, while the private key consists of: (i) the quantum walk, (ii) the number of steps of the walk and (iii) the initial state. 
Regarding the use of quantum walks in cryptography, see for example the recent work by Yan {\it et al}~\cite{yan:pan:sun:xu:15}.

The paper is organized as follows: 
in the next section we start with basic definitions and notation used throughout the paper. 
In Section~\ref{sec:pkscheme} we present our public-key cryptographic system and discuss its security and efficiency.
In Section~\ref{sec:conclusions} we summarize the results obtained and point some future research direction.

\section{Quantum walks}
The concept of random walk was introduced to describe the behavior of a ``walker'' over a path who, at each step, can choose to follow one of the possible directions with a certain {\em a priori} fixed probability. It was shown to be very useful in computer science (sampling massive online graphs, image segmentation, estimating the size of the World Wide Web, wireless networking, etc.), physics (modeling the Brownian motion, studying polymers, etc.), and many other fields of research (financial economics, medicine and biology, psychology, etc.).

Quantum walks are the quantum counterpart of random walks and were first considered in ~\cite{aha:dav:zag:93}. Unlike the classical case, in which the state of the walker is described by a probability distribution over the allowed positions, in the quantum scenario the state of the walker is given by a superposition of positions. Quantum walks are a successful tool in algorithmic theory since they provide an exponential speedup over classical computations for several oracular problems~\cite{chi:etal:03,kem:03,amb:03} and polynomial speedup for many important classical solutions of problems such as determining whether all the elements of a list are distinct~\cite{san:08}, and in search problems in general~\cite{por:13}.
Moreover, they are an important computational primitive, since they permit universal quantum computation~\cite{lov:coo:eve:tre:ken:10,chi:gos:web:13}. For practical implementations, see for example the recent book~\cite{man:wan:13}.
One can study different types of quantum walks, determined by their time evolution (discrete- vs continuous-time), underlying position space and their topology (walks on the line, lattice, circle, graphs, etc.). In this work we use discrete-time quantum walks on the circle.

\label{sec:preliminaries}
\subsection{One-dimensional discrete-time quantum walks}

\subsubsection{Basic dynamics}
In a discrete-time quantum walk on an infinite line, we consider the movement of a walker along discrete positions on  it, labeled $i \in \mathbb{Z}$. At each step the particle moves to the left or to the right, depending on the state of the internal degree of freedom, commonly called the coin state. The position and the coin state of the particle are from the Hilbert spaces 
$\mathcal{H}_p =\span \{\Ket{i}: i \in \mathbb{Z}\}$ 
and 
$\mathcal{H}_c = \span\{ \ket{R}, \Ket{L} \}$, respectively.

The evolution of the system at each step of the walk is described by the unitary operator:
\begin{equation}
\hat{U} = \hat{S}\left(\hat{I}_p \otimes \hat{U}_c \right).
\end{equation}
In the above expression $\hat{I}_p$ is the identity operator on $\mathcal{H}_p$, $\hat{S}$ is the shift operator
\begin{equation}
\hat{S} = \sum_{i\in\mathbb{Z}}
			\left(\Ket{i+1}\Bra{i}  \otimes \Ket{R} \Bra{R} + 
					\Ket{i-1}\Bra{i} \otimes \Ket{L} \Bra{L}\right),
\end{equation}
and $\hat{U}_c \in U(2)$ is the coin operator acting on $\mathcal{H}_c$. 
The general expression for $\hat{U}_c$ is:
\begin{equation}
\label{coin}
\hat{U}_c = \hat{U}_c(\theta, \xi, \zeta) =
			\left[
				\begin{array}{cc}
					e^{i\xi}
					\cos\theta & e^{i\zeta}
					\sin\theta\\
					-e^{-i\zeta}
					\sin\theta & e^{-i\xi}
					\cos\theta
				\end{array}
			\right].
\end{equation}
 
\subsubsection{Shift operator on the circle}

To simulate the walk on a circle, one could either identify positions $-N$ and $N$ of the above discussed line, or connect the two, thus altering the corresponding shift operator. In the former, the circle has an even number of positions ($2N$), while in the latter it has an odd number of positions ($2N+1$). 

Without loss of generality, we consider the position Hilbert space 
$\mathcal{H}_p =\span\{\Ket{i}:i\in\{0,\dots,N-1\}\}$. 
The general expression for a shift operator on a circle with $N$ positions is:
\begin{eqnarray}
\label{shift}
\hat{S} &=&  \sum_{i= 0}^{N-1} \Big(\Ket{i+1 \ppmod {N}}\Bra{i}  \otimes \Ket{R} \Bra{R} \nonumber\\
		&& \hspace*{10mm}+  \Ket{i-1 \ppmod N}\Bra{i}  \otimes \Ket{L} \Bra{L} \Big)\\
		&=& \hat{T}_1 \otimes \Ket{R} \Bra{R} + \hat{T}_{-1} \otimes \Ket{L} \Bra{L}\nonumber
\end{eqnarray}
where 
\begin{equation}
\hat{T}_m = \sum_{i=0}^{N-1}\ket{ i+m \ppmod N}\Bra{i}
\end{equation}
is the $m-$position translation operator.

\section{Public-key encryption based on discrete-time quantum walks}
\label{sec:pkscheme}
We now present the public-key cryptographic system using quantum walks on a circle. To generate the public key, we use a discrete number of possible walks $\hat{U}_k = \hat{S} [ \hat{I} \otimes \hat{U}_c (\theta_k, \xi_k, \zeta_k)]$, with $\theta_k = \xi_k = \zeta_k = k \frac{2\pi}{d}$, $k~\in~\mathcal I=\{1,2,\ldots,d\}$ and $d \in \mathbb N$.

\begin{protocol}[Public-key encryption scheme]\
\label{prot:pk} 
\begin{description}
\item[\hspace{3mm}{\em Inputs for the protocol}]\
	\begin{itemize}
		\item Message to transfer: 

			$m\in \{0, \dots, 2^n-1 \}$, i.e., a message of at most $n$ bits;

		\item Secret key $SK = (\hat U_k, t, l)$ where:

			$\hat{U}_k$ with $k~\in~\mathcal I=\{1,2,\ldots,d\}$, 
			$t~\in~\mathcal{T} = \{ t_0, \dots, t_{max} \} \subset \mathbb{N}$ and 
			$l \in \{ 0, \dots, 2^n-~1\} $.
	\end{itemize}

\item[\hspace{3mm} {\em Public-key generation}]\
	\begin{itemize}
	\item Alice chooses uniformly at random $l \in \{0, \dots ,2^n-1 \} $ and 
	$s \in\{ L,R\}$, and generates the initial state $\Ket{l}\Ket{s}$;
	
	\item Then she chooses, also uniformly at random, the walk 
	$\hat{U}_k = \hat{S} (\hat{I}_p \otimes \hat U_{c})$ 
	and the number of steps $t\in\mathcal T$;   
	
	\item Finally, she generates the public key:
	\begin{equation}
	\Ket{\psi_{PK}} 	= \hat{U}_k^t\Ket{l}\Ket{s} 
				= \left[\hat{S} (\hat{I}_p \otimes \hat U_{c})\right]^t \Ket{l}\Ket{s}.
	\end{equation}
	\end{itemize}

\item[\hspace{3mm} {\em Message Encryption}]\
	\begin{itemize}
	
	\item Bob obtains Alice's public key $\Ket{\psi_{PK}}$;
	
	\item He encrypts $m$ by applying spatial translation to obtain: 
		\begin{equation}
		\Ket{\psi(m)} = (\hat{T}_m\otimes \hat{I}_c) \Ket{\psi_{PK}};
		\end{equation} 
		
	\item Bob sends $\Ket{\psi(m)}$ to Alice.
	\end{itemize}
	
	\item[\hspace{3mm} {\em Message Decryption}]\
	\begin{itemize}
	\item Alice applies $\hat{U}_k^{-t}$ to the state $\Ket{\psi(m)}$;

	\item She performs the measurement 
	\begin{equation}
	\hat{M} = \sum_{i}\Ket{i} \Bra{i}\otimes \hat{I}_c
	\end{equation} 
	and obtains the result $m'$. The message sent by Bob is $m= m' - l\!\!\!\!\pmod N$.
\end{itemize}
\end{description}
\end{protocol}

\subsection{Correctness of the protocol}

\begin{proposition}
The above protocol  is correct, in the sense that if Alice and Bob follow it, and no third party intervenes during its execution, at the end of the decryption phase Alice recovers the message sent by Bob with probability $1$.
\end{proposition}

\begin{proof}
The correctness of the protocol when Alice and Bob follow the prescribed steps is a direct consequence of the fact that the quantum walk $\hat{U}_k^t$ commutes with any translation $\hat{T}_m$. Thus, the state of the system before the final step of the decryption phase (measurement), is: 
\begin{eqnarray}
\Ket {\psi_f}	&=& \hat{U}_k^{-t}\Ket{\psi(m)} \nonumber\\
			&=& \hat{U}_k^{-t} (\hat{T}_m \otimes \hat{I}_c) \hat{U}_k^t \Ket{l}\Ket{s} \\
			&=& (\hat{T}_m \otimes \hat{I}_c) \Ket{l}\Ket{s} \nonumber\\
			&=& \Ket{l+m  \ppmod N}\Ket{s}.\nonumber
\end{eqnarray}
Hence, upon measuring $\hat{M}$ and obtaining $m' = l+m \!\!\pmod N$, the last modular operation performed in the last step of the {\em Message Decryption} reveals that the decrypted message is indeed $m$. 
\end{proof}

In the previous Proposition the crucial point is the fact that the quantum walk commutes with the translation. Below, we prove in detail that $\hat{U}_k^t$ and $(\hat{T}_m \otimes \hat{I}_c)$ commute. 

\begin{lemma}
Let $N\geq 2^n$ where $n$ is a fixed integer. 
Let $\hat{U}_k^t$ be a quantum walk from Protocol~\ref{prot:pk} and let $\hat{T}_m$ denote the translation operator for $m$ positions modulo $N$. Then $\hat{U}_k^t$ and $(\hat{T}_m\otimes \hat{I}_c)$ commute.
\end{lemma}

\begin{proof}
Notice that the action of any $\hat U_k$ used in Protocol~\ref{prot:pk} can be written as:
\begin{eqnarray}
\hat U_k\Ket{l}\Ket{s} = \alpha_{L(s)} \Ket{l-1} \Ket{L} +\alpha_{R(s)} \Ket{l+1} \Ket{R},
\end{eqnarray}
where $\Ket{L}$ and $\Ket{R}$ are the orthogonal coin states and 
$\alpha_{L/R(s)}$ is the probability amplitude to find the walker in position $l-1$ or $l+1$, depending on its spin.
Notice also that $\hat{T}_m$ is defined as:
\begin{eqnarray}
\hat{T}_m \Ket{l} = \Ket{l+m \ppmod N}.
\end{eqnarray}
Then, for any element of the form $\Ket{l}\Ket{s}$ we have:
\begin{eqnarray}
(\hat{T}_m\otimes \hat{I}_c) \hat{U}_k  \Ket{l}\Ket{s}
	&=&(\hat{T}_m\otimes \hat{I}_c) [ \alpha_{L(s)} \Ket{l-1} \Ket{L} + \alpha_{R(s)} \Ket{l+1}\Ket{R} ]\nonumber\\
	&=& \alpha_{L(s)} \Ket{l-1+ m  \ppmod N} \Ket{L}\\ && + \alpha_{R(s)} \Ket{l+1 + m  \ppmod N} \Ket{R}.\nonumber
\end{eqnarray}
On the other hand, we also have:
\begin{eqnarray}
	\hat{U}_k(\hat{T}_m\otimes \hat{I}_c) \Ket{l}\Ket{s}
		&=& \hat{U}_k \Ket{l + m  \ppmod N}\Ket{s}\nonumber\\
		&=& \alpha_{L(s)} \Ket{l-1+ m  \ppmod N} \Ket{L} \\ 
		& &+ \alpha_{R(s)} \Ket{l+1 + m  \ppmod N} \Ket{R}.\nonumber
\end{eqnarray}
\end{proof}
Observe that this lemma can be extended to more general shift operations, which  allow jumps for two or more positions, or leave position state unchanged, depending on the coin state.

\subsection{Security of the protocol}

The protocol consists of two phases. In the first, Alice sends a public key $\Ket{\psi_{PK}}$ to Bob. In the second, upon encrypting the message $m$, Bob sends back the state $\Ket{\psi(m)}$ to Alice. Therefore, one has to show the security of the secret key during the first phase and the security of the message during the second phase.

Our proof of the security is based on Holevo's Theorem, that bounds the amount of classical information that an eavesdropper can retrieve from a given quantum mixed state by means of a POVM measurement.

Let us denote by $\hat{\rho}_{PK}$ the mixed state of the public key, as perceived by Eve, who does not know \textit{a priori} the secret key $SK$ chosen by Alice. 
Even if Eve were to know $\hat U_k$ and $t$, $\hat{\rho}_{PK}$ is completely mixed:
\begin{eqnarray}
\hat{\rho}_{PK} 
	&=& \hat{U}_k^t\left[\frac{1}{2^{n+1}} \sum_{l = 0}^{2^n-1} \sum_{s\in\{ L,R\}} \Ket{l} \Bra{l} \otimes \Ket{s} \Bra{s}\right](\hat{U}_k^t)^{\dagger}\nonumber\\
	&=& \hat{U}_k^t \left(\frac{1}{2^{n+1}} \hat{I}_p \otimes \hat{I}_c\right) (\hat{U}_k^t)^{\dagger}
	\\
	&=& \frac{1}{2^{n+1}} \hat{I}_p \otimes \hat{I}_c\nonumber.
\end{eqnarray}
Assuming that Eve performs a measurement on $\hat{\rho}_{PK}$, Holevo's Theorem implies that the mutual information $I(SK,E)$ between the secret key $SK$ and her inference $E$ is bounded from above by the Von Neumann entropy of this state:
\begin{eqnarray}
I(SK,E) \leq S(\hat{\rho}_{PK})=-\Tr(\hat{\rho}_{PK} \log\hat{\rho}_{PK} )=n+1.
\end{eqnarray}

To conclude that the protocol is secure we have to show that the mutual information is very small compared to the Shannon entropy of the secret key.
Indeed, the Shannon entropy of the secret key depends on the probability to choose $\hat U_k$, $t$ and the initial state $\ket l\ket s$.
In the following we denote by $p_k$ the probability to choose $\hat{U}_k$ from the set $\left\{\hat{U}_k | k\in \mathcal I= \{ 1,2,\dots,d\} \right\}$, by $p_t$ the probability to run the walk for $t$ steps, with $t \in \mathcal T=\{ t_0, \dots, t_{max} \}$, and by $p_{l,s}$ the probability to choose $\Ket{l}\ket s$ as the initial state, where $l\in\{0,1,\ldots,2^n-1\}$ and $s\in\{L,R\}$.
Since this choices are random and independent, the probability of a certain secret key $SK$ is given by:
\begin{equation}
p_{SK}=p_k \; p_t \; p_{l,s}=\frac{1}{d\; |\mathcal T|\; 2^{n+1}},
\end{equation}
where $|\mathcal T|$ is the cardinality of $\mathcal T$.

The above probability distributions are uniform, so the Shannon entropy of the secret key is:
\begin{eqnarray}
H(p_{SK})&=&
	-\sum_{k\in \mathcal I}\sum_{t\in \mathcal T}\sum_{l=0}^{2^{n}-1} \sum_{s\in\{L,R\}} p_k \; p_t \; p_{l,s}\log_2(p_k \; p_t \; p_{l,s})\nonumber\\[2mm]
	&=& \log_2(d \; |\mathcal T| \; 2^{n+1}) \\[2mm]
	&=& \log_2(d \; |\mathcal T| ) + n + 1.\nonumber
\end{eqnarray}
Thus, we have:
\begin{equation}
I(SK,E)\leq S(\hat{\rho}_{PK})<H(p_{SK}).
\end{equation}
With the appropriate choice of $|\mathcal T|$ and $d$, e.g.,  $|\mathcal T|, \log d \approx \poly(n)$, for sufficiently large $n$, the Shannon entropy of the secret  key has a polynomial overhead over the von Neumann entropy of the public key as seen by Eve,  \begin{equation}H(p_{SK}) - S(\hat{\rho}_{PK}) = \log_2(d \; |\mathcal T| ) \approx \poly(n).
\end{equation}
This way, upon obtaining the maximal possible information about the secret key, given by $S(\hat{\rho}_{PK})$, Eve's uncertainty (in the number of bits) of the $SK$ is still polynomial in $n$, i.e., the number of keys consistent with the information she has is exponential in $n$. We note that the choice of $d \approx \exp(n)$ 
secures the secrecy of the encrypted message, while $|\mathcal T|  \approx \poly(n)$ was chosen to maintain the protocol's efficiency, discussed in the next section.

For the rest of this section we discuss the security of the message $m$ during the second phase of the protocol, when Bob sends the encrypted message $\Ket{\psi(m)} = (\hat{T}_m\otimes \hat{I}_c) \Ket{\psi_{PK}}$ to Alice. Without knowing the secret key, the state perceived by Eve is still a complete mixture: 
\begin{eqnarray}
\hat{\rho}_E 
	&=&  (\hat{T}_{m}\otimes \hat{I}_c) \left(\frac{1}{2^{n+1}}  \hat{I}_p \otimes \hat{I}_c\right) (\hat {T}_m\otimes \hat{I}_c)^{\dagger}= \frac{1}{2^{n+1}}\hat{I}_p \otimes \hat{I}_c.
\end{eqnarray}
The most that Eve can learn is the very quantum state $\Ket{\psi(m)}$ (although, as proven above, even that is impossible, unless with negligible probability). Nevertheless, without knowing the secret key, this information is not enough for Eve to infer the message encrypted by Bob. This is a simple consequence of the fact that for each allowed encryption state, there exists a suitably chosen secret key that can decrypt {\em any} message $m$. Indeed, a state $\Ket{\psi(m)}$ that for the secret key $SK = (\hat U_k, t, l)$ corresponds to the message $m$, for the secret key $SK' = (\hat U_k, t, l - \Delta l)$ corresponds to the message $m + \Delta l$ (below, the subscripts $SK$ and $SK'$ explicitly denote the secret key used to encrypt the corresponding messages $m$ and $m + \Delta l$, respectively):
\begin{eqnarray}
\Ket{\psi(m)}_{SK} 	& = & (\hat{T}_m\otimes \hat{I}_c) \Ket{\psi_{PK}} \nonumber \\
			& = & (\hat{T}_{m}\otimes \hat{I}_c) \hat{U}_k^t  (\hat{T}_{\Delta l}\otimes \hat{I}_c)\Ket{l - \Delta l}\Ket{s} \\
			& = & (\hat{T}_{m + \Delta l}\otimes \hat{I}_c) \hat{U}_k^t  \Ket{l - \Delta l}\Ket{s} \nonumber \\
			& = & \Ket{\psi(m + \Delta l)}_{SK'}. \nonumber 
\end{eqnarray}

%
%

\subsection{Efficiency of the protocol}


In this section we show the efficiency of the proposed protocol, i.e. the overall time $\tau$ required for its execution (public-key generation, message encryption and message decryption) scales polynomially with the length $n$ of the message. 

The public-key generation as well as the message decryption are efficient procedures, since the quantum walks used to perform them are efficient.
Indeed, denoting by $\Delta\tau_w$ the time required for a single step $\hat U$ of the walk, the full walk $\hat U^t$ is completed in time $\tau=t\cdot\Delta\tau_w$. In the previous section we took $t \approx \poly(n)$ for security purposes, a choice which is also adequate for the efficiency of the quantum walk: the time required to perform the walk is polynomial in $n$.

In addition to this, for the overall protocol to be efficient, the message encryption, given by the translation operator  $\hat{T}_m$, has to be efficient as well. It might seem at first that the encryption of the message is not efficient, as it requires $\O(2^n)$ single-position translations, $\hat{T}_m=(\hat{T}_1)^m$. Below, we show that this is not necessarily a non-efficient procedure, i.e., various practical implementations of $\hat{T}_m$ are indeed efficient.

In case the system that performs the quantum walk consists of $n+1$ qubits ($n$ carrying the position of the walker plus the coin one), such that the states of the computational basis encode different positions (see for example~\cite{rya:laf:boi;laf:05}), the translation operator $\hat{T}_m$ is nothing but the addition by $m$, which is an efficient operation for a quantum computer.
Alternatively, in those cases of physical realizations in which different position states $\ket i$ are given by distinct spatial positions (see for example implementations based on integrated photonics~\cite{san:etal:12}), Bob can simply re-label the positions on the device that carries the quantum state of the public key, i.e. $i \rightarrow i-m$, which is also efficient, as he can do it in parallel at the same time for all position states.

\section{Conclusions}
\label{sec:conclusions}
In this paper we presented a quantum public-key cryptographic system based on quantum walks. Unlike the recently proposed protocol~\cite{nik:08} that uses single-qubit rotations to generate the public key, in our scheme the execution of a quantum walk, in general, results in entangled quantum states as public keys, thus increasing the practical security of our scheme (an eavesdropper has to, in general, perform more complex operations to extract information from entangled rather than from product states). Using Holevo's theorem, we proved the protocol's security. We also analyzed the complexity of our public-key generation and message encryption/decryption and showed their efficiency, i.e., that the complexity of our protocol scales polynomially with the size of the message.

Future research includes designing other security protocols based on the use of quantum walks, such as oblivious transfer (along the lines of the recently proposed protocol presented in~\cite{pmat:npaunkovic:jrodr:asouto:14}), commitment schemes and other privacy functionalities.

\section*{Acknowledgements}
This work was partially supported, under the CV-Quantum initiative at IT, by FCT PEst-OE/EEI/LA0008/2013 and UID/EEA/50008/2013, namely via the FCT CaPri initiative. 
We also acknowledge bilateral scientific co-operation between Portugal and Serbia through the project ``Noise and measurement errors in multi-party quantum security protocols'', no. 451-03-01765/2014-09/04.
C. Vlachou acknowledges the support of DP-PMI and FCT (Portugal) through the scholarship PD/BD/52652/2014.
J. Rodrigues and A. Souto acknowledge, respectively, the FCT grants SFRH/BD/75085/2010 and SFRH/BPD/76231/2011. 
N. Paunkovi\'c acknowledges discussions with D. Arsenovi\'c, D. Popovi\'c and S. Prvanovi\'c.
\bibliographystyle{unsrt}
\bibliography{quantum}

\end{document}